\documentclass[preprint,12pt]{elsarticle}

\usepackage{graphicx}

\usepackage{amsmath}
\usepackage{amsfonts}
\usepackage{color}

\begin{document}

\newcommand{\hide}[1]{}
\newcommand{\tbox}[1]{\mbox{\tiny #1}}
\newcommand{\half}{\mbox{\small $\frac{1}{2}$}}
\newcommand{\sinc}{\mbox{sinc}}
\newcommand{\const}{\mbox{const}}
\newcommand{\trc}{\mbox{trace}}
\newcommand{\intt}{\int\!\!\!\!\int }
\newcommand{\ointt}{\int\!\!\!\!\int\!\!\!\!\!\circ\ }
\newcommand{\eexp}{\mbox{e}^}
\newcommand{\bra}{\left\langle}
\newcommand{\ket}{\right\rangle}
\newcommand{\EPS} {\mbox{\LARGE $\epsilon$}}
\newcommand{\ar}{\mathsf r}
\newcommand{\im}{\mbox{Im}}
\newcommand{\re}{\mbox{Re}}
\newcommand{\bmsf}[1]{\bm{\mathsf{#1}}}
\newcommand{\mpg}[2][1.0\hsize]{\begin{minipage}[b]{#1}{#2}\end{minipage}}

\newcommand{\CC}{\mathbb{C}}
\newcommand{\NN}{\mathbb{N}}
\newcommand{\PP}{\mathbb{P}}
\newcommand{\RR}{\mathbb{R}}
\newcommand{\QQ}{\mathbb{Q}}
\newcommand{\ZZ}{\mathbb{Z}}

\renewcommand{\a}{\alpha}
\renewcommand{\b}{\beta}
\renewcommand{\d}{\delta}
\newcommand{\D}{\Delta}
\newcommand{\g}{\gamma}
\newcommand{\G}{\Gamma}
\renewcommand{\th}{\theta}
\renewcommand{\l}{\lambda}
\renewcommand{\L}{\Lambda}
\renewcommand{\O}{\Omega}
\newcommand{\s}{\sigma}

\newtheorem{theorem}{Theorem}
\newtheorem{acknowledgement}[theorem]{Acknowledgement}
\newtheorem{algorithm}[theorem]{Algorithm}
\newtheorem{axiom}[theorem]{Axiom}
\newtheorem{claim}[theorem]{Claim}
\newtheorem{conclusion}[theorem]{Conclusion}
\newtheorem{condition}[theorem]{Condition}
\newtheorem{conjecture}[theorem]{Conjecture}
\newtheorem{corollary}[theorem]{Corollary}
\newtheorem{criterion}[theorem]{Criterion}
\newtheorem{definition}[theorem]{Definition}
\newtheorem{example}[theorem]{Example}
\newtheorem{exercise}[theorem]{Exercise}
\newtheorem{lemma}[theorem]{Lemma}
\newtheorem{notation}[theorem]{Notation}
\newtheorem{problem}[theorem]{Problem}
\newtheorem{proposition}[theorem]{Proposition}
\newtheorem{remark}[theorem]{Remark}
\newtheorem{solution}[theorem]{Solution}
\newtheorem{summary}[theorem]{Summary}
\newenvironment{proof}[1][Proof]{\noindent\textbf{#1.} }{\ \rule{0.5em}{0.5em}}


\journal{Applied Mathematics and Computation}

\begin{frontmatter}



\title{Computational and analytical studies of the Randi\'c index in Erd\"os-R\'{e}nyi models}


\author[BUAP]{C. T. Mart{\'\i}nez-Mart{\'\i}nez}
\author[BUAP]{J. A. M\'endez-Berm\'udez\corref{mycorrespondingauthor}}
\cortext[mycorrespondingauthor]{Corresponding author} 
\ead{jmendezb@ifuap.buap.mx}
\author[UIII]{Jos\'e M. Rodr\'{\i}guez}
\author[BUAP,UAGRO]{ Jos\'e M. Sigarreta }
\address[BUAP]{Instituto de F\'{\i}sica, Benem\'erita Universidad Aut\'onoma de Puebla,
Apartado Postal J-48, Puebla 72570, Mexico}
\address[UIII]{Universidad Carlos III de Madrid, Departamento de Matem\'aticas, Avenida de la Universidad
30, 28911 Legan\'es, Madrid, Spain}
\address[UAGRO]{Universidad Aut\'onoma de Guerrero, Centro Acapulco CP 39610,
Acapulco de Ju\'arez, Guerrero, Mexico}

\begin{abstract}

In this work we perform computational and analytical studies of the 
Randi\'c index $R(G)$ in Erd\"os-R\'{e}nyi 
models $G(n,p)$ characterized by $n$ vertices connected independently with probability $p \in (0,1)$.
First, from a detailed scaling analysis, we show that $\bra \overline{R}(G) \ket = \bra R(G)\ket/(n/2)$ scales 
with the product $\xi\approx np$, so we can define three regimes: 
a regime of mostly isolated vertices when $\xi < 0.01$ ($R(G)\approx 0$),
a transition regime for $0.01 < \xi < 10$ (where $0<R(G)< n/2$), and
a regime of almost complete graphs for $\xi > 10$ ($R(G)\approx n/2$). 
Then, motivated by the scaling of $\bra \overline{R}(G) \ket$, we analytically 
(i) obtain new relations connecting $R(G)$ with other topological indices and characterize graphs which are 
extremal with respect to the relations obtained and
(ii) apply these results in order to obtain inequalities on $R(G)$ for graphs in Erd\"os-R\'{e}nyi models.
\end{abstract}


\end{frontmatter}


\section{Introduction}

The interest in topological indices lies in the fact that they synthesize some of the fundamental properties of a molecule into a single value. With this in mind, 
several topological indices have been studied so far; it is worth noting the seminal work by 
Wiener (see \cite{Wi}) in which he used the distances of a  chemical graph in order to model  properties of alkanes.

The \emph{Randi\'c connectivity index} was defined in \cite{R} as
\begin{equation}
R(G) = \sum_{uv\in E(G)} \frac1{\sqrt{d_u d_v}}\, ,
\end{equation}
where $uv$ denotes the edge of the graph $G$, and $d_u$ is the degree of the 
vertex $u$. Indeed, there are lots of works dealing with this index (see, e.g., 
\cite{GF,LG,LS}).

In \cite{LZheng,LZhao,MN}, the 
\emph{first and second variable Zagreb indices} are defined as
$$
M_1^{\a}(G) = \sum_{u\in V(G)} d_u^{\a} \ ,
\qquad
M_2^{\a}(G) = \sum_{uv\in E(G)} (d_u d_v)^\a ,
$$
with $\a \in \RR$. The concept of \emph{variable molecular descriptors} was proposed as a new way of characterizing heteroatoms in 
molecules (see \cite{R2,R3}). The essential idea is that the
variables are determined during the regression; this allows to make the standard error of the estimate for a particular 
property (targeted in the study) as small as possible (see, e.g., \cite{MN}). The second variable Zagreb index is used in the structure-boiling point modeling of benzenoid hydrocarbons \cite{NMTJ}.

The \emph{general sum-connectivity index} was defined  in \cite{ZT2} as
$$
\chi_{\a}(G) = \sum_{uv\in E(G)} (d_u+ d_v)^\a .
$$
Some relations of these indices are  studied in (\cite{RS4}).

\medskip

In addition to the multiple applications of the Randi\'c index in physical chemistry, this index has found several 
applications in other research areas and topics, such as information theory \cite{GFK18}, network similarity 
\cite{NJ03}, protein alignment \cite{R15}, network heterogeneity \cite{E10}, and network robustness \cite{MMR17}. 
Moreover, in \cite{CDE15} the concept of graph entropy for weighted graphs was introduced, especially the 
Randi\'c weights. 

\medskip

We want to recall that graphs have been widely used to study the properties of highly complex systems.
Among them we can mention biological, social, and technological networks \cite{Costa011:AP,B13}. 
Moreover, graphs can be classified as deterministic (regular and fractal) or disordered (random) \cite{Mulken2011}. 
Deterministic graphs follow specific construction rules, while in random graphs the parameters take fixed values 
but the graph itself has a random structure.
In the later case a statistical study of graph ensembles with the same average properties must be performed, since
the analysis of a single random graph is meaningless. There are well-known models
of random graphs in the literature \cite{N10,Boccaletti06}, presumably the most popular are:
the Erd\"os-R\'{e}nyi model of random graphs, scale-free networks (introduced by Barab\'{a}si and Albert), 
and small-world networks (introduced by Watts and Strogatz). These three models have been extensively 
used to represent the organization of real-world complex systems (such as power grids or the Internet) through 
their underlying network structure \cite{Costa011:AP,N10,Boccaletti06}. 
\medskip

Although random graph models are not able to predict some properties observed in real-world networks, 
such as nonvanishing clustering coefficient and power-law degree distributions \cite{Boccaletti06}, they have been 
deeply studied theoretically (e.g.~\cite{Bollobas98}). In fact, several important results, such as the emergence of 
percolation, are analytically accesible from Erd\"os-R\'{e}nyi graphs \cite{N10,Bollobas98}.
Thus, here we consider Erd\"os-R\'{e}nyi random graphs, which were proposed by Solomonoff and 
Rapoport \cite{SR51} and investigated later in great detail by Erd\H{o}s and R\'enyi \cite{ER59,ER60}.

\medskip

This work is organized as follows.
First, in Sec.~\ref{Sec:Num} we perform a detailed scaling analysis of the average Randi\'c index to find its
{\it universal} parameter, i.e., the parameter that statistically fixes the average value of $R(G)$. 
Then, in Sec.~\ref{Sec:Ana}, we analytically 
(i) obtain new relations connecting $R(G)$ with other topological indices and
(ii) apply these results in order to obtain inequalities on $R(G)$ for graphs in Erd\"os-R\'{e}nyi models.

\section{Scaling analysis of the Randi\'c index on Erd\"os-R\'{e}nyi graphs}
\label{Sec:Num}

We start with a computational (and statistical) study of the Randi\'c index on Erd\"os-R\'{e}nyi graphs.
We consider random graphs $G$ from the standard Erd\"os-R\'{e}nyi model $G(n, p)$, i.e., $G$ has $n$ vertices 
and each edge appears independently with probability $p \in (0,1)$.

In Fig.~\ref{Fig1}(a) we show the average Randi\'c index $\bra R(G) \ket$ as a function of the probability $p$ of 
Erd\"os-R\'{e}nyi graphs $G(n, p)$ of several orders $n$. Here, the average $\bra \cdot \ket$ is computed over 
2000 random graphs $G(n, p)$.
We observe that the curves of $\bra R(G) \ket$, for all the values of $n$ considered here, have a very similar 
shape as a function of $p$: $\bra R(G) \ket$ shows a smooth transition (in log scale) from 
zero to $n/2$ when $p$ increases from zero (isolated vertices) to one (complete graphs). 
Note that $n/2$ is the maximal value that $R(G)$ can take.

Now, to ease our analysis, in Fig.~\ref{Fig1}(b) we present again $\bra R(G) \ket$ but now normalized to $n/2$:
\begin{equation}
\bra \overline{R}(G) \ket = \frac{\bra R(G) \ket}{n/2} .
\end{equation}
From this figure we can clearly see that the main effect of increasing $n$ is the 
displacement of the curves $\bra \overline{R}(G) \ket$ vs.~$p$ to the left on the $p$-axis. 
Moreover, the fact that these curves, plotted in semi-log scale, are shifted the same amount on the $p$-axis when 
doubling $n$ make us anticipate the existence of a scaling parameter that depends on $n$. In order to search for that 
scaling parameter we first establish a measure to characterize the position of the curves $\bra \overline{R}(G) \ket$ on 
the $p$-axis: We choose the value of $p$, that we label as $p^*$, for which 
$\bra \overline{R}(G) \ket \approx 0.5$; see the dashed line in Fig.~\ref{Fig1}(b).
Notice that $p^*$ locates the transition point from isolated vertices to
complete Erd\"os-R\'{e}nyi graphs of size $n$.

\begin{figure}
 \centering
\includegraphics[width=0.95\textwidth]{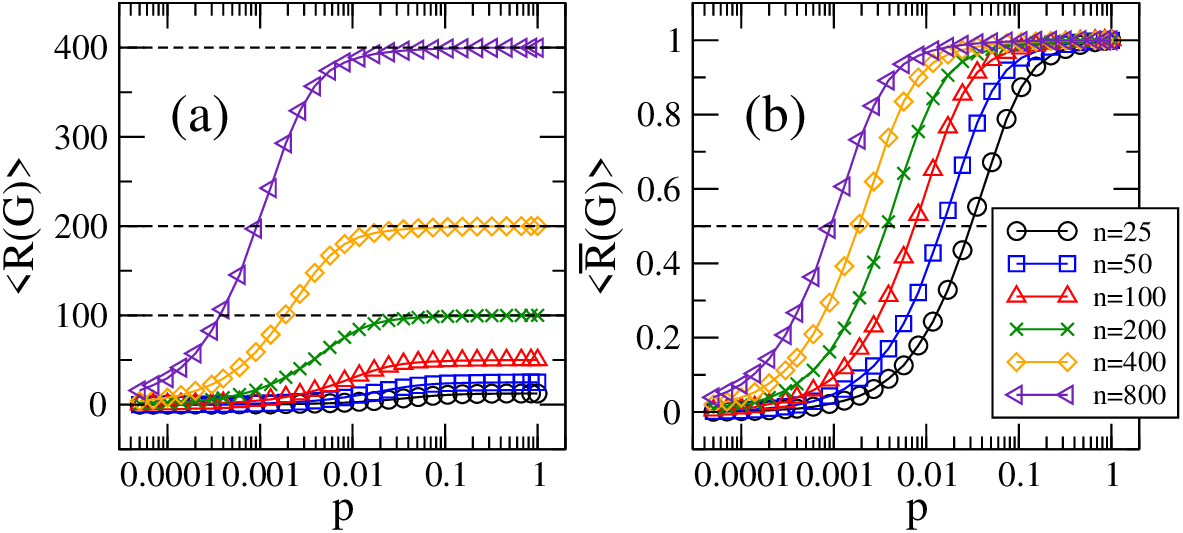}
\caption{(a) Average Randi\'c index $\bra R(G) \ket$ as a function of the probability $p$ of Erd\"os-R\'{e}nyi 
graphs $G(n,p)$ of different sizes $n\in[25,800]$. (b) $\bra R(G) \ket$ normalized to $n/2$, $\bra \overline{R}(G) \ket$,
as a function of $p$. Dashed lines in (a) indicate the values of $n/2$ for $n\in[200,800]$. The dashed line in (b) 
indicates $\bra \overline{R}(G) \ket = 0.5$, used to define $p^*$. Each symbol was computed by averaging 
over 2000 random graphs $G(n, p)$.}
\label{Fig1}
\end{figure}
\begin{figure}
 \centering
  \includegraphics[width=0.95\textwidth]{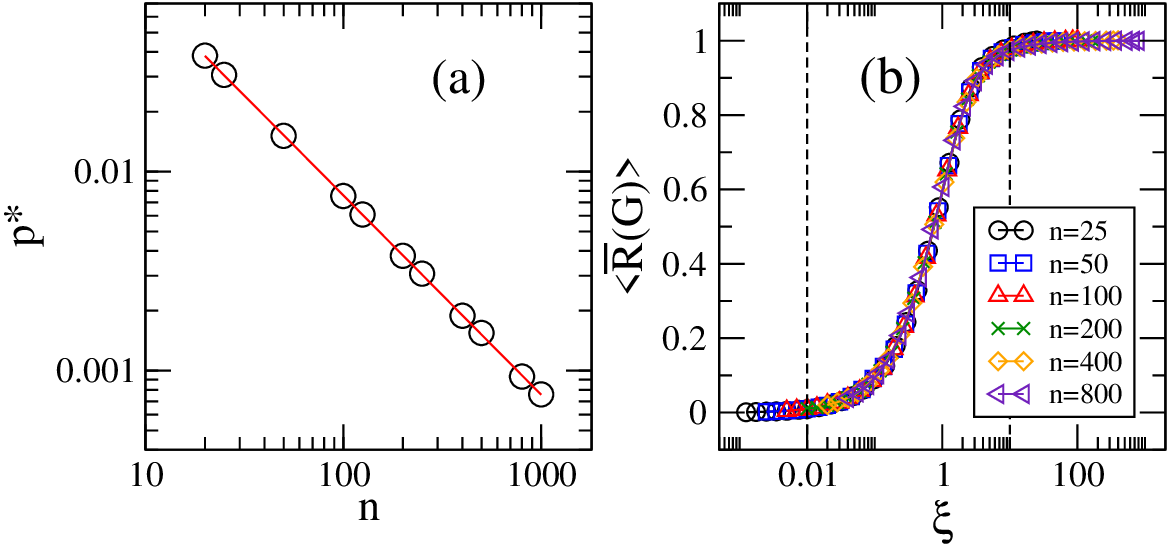}
\caption{(a) $p^*$ (defined as the value of $p$ for which $\bra \overline{R}(G) \ket \approx 0.5$) as a function 
of the graph size $n$. The red line is the fitting of Eq.~(\ref{scaling}) to the the data with fitting parameters
$\mathcal{C} = 0.76775$ and $\delta = -1.0021$. (b) $\bra \overline{R}(G) \ket$ as a function of 
$\xi$. Vertical dashed lines in (b) indicate: The regime of mostly isolated vertices ($\xi < 0.01$), the transition 
regime ($0.01 < \xi < 10$), and the regime of almost complete graphs ($\xi > 10$).}
\label{Fig2}
\end{figure}

Then, in Fig.~\ref{Fig2}(a) we plot $p^*$ versus $n$. The linear trend of the data (in log-log scale) 
in Fig.~\ref{Fig2}(a) suggests the power-law
\begin{equation}
\label{scaling}
p^* = \mathcal{C} n^\delta .
\end{equation}

In fact, Eq.~(\ref{scaling}) provides an excellent fitting to the data with $\mathcal{C} \approx 0.77$ and 
$\delta \approx -1$. Therefore, by plotting again the curves of $\bra \overline{R}(G) \ket$ now as a function of 
the probability $p$ divided by $p^*$, 
\begin{equation}
\label{xi}
\xi \equiv \frac{p}{p^*} \propto \frac{p}{n^\delta} \approx \frac{p}{n^{-1}} = np \ ,
\end{equation} 
we observe that curves for different graph sizes $n$ collapse on top of a single {\it universal} curve, see 
Fig.~\ref{Fig2}(b). This means that once the product $np$ is fixed, the average Randi\'c index on Erd\"os-R\'{e}nyi 
graphs is also fixed. This statement is in accordance with the results reported in \cite{MAMRP15,MAM13}, 
where the spectral and transport properties of Erd\"os-R\'{e}nyi graphs where shown to be universal for the 
scaling parameter $np$, see also \cite{GAC18,MM19,TFM19}.

Additionally, from our previous experience, see e.g., \cite{MAMRP15,MAM13,GAC18,MM19,TFM19}, we expect 
that other quantities related to $R(G)$ will also be scaled with $\xi$. Indeed, we validate this conjecture by 
analyzing the energy $E(n,p)$ of the Erd\"os-R\'{e}nyi graphs $G(n,p)$ defined as \cite{RS05,BGGC10}
\begin{equation}
\label{energy}
E(n,p) = \sum_{i=1}^n | e_{i} | \ ,
\end{equation}
where $e_{i}$ are the eigenvalues of the corresponding Randi\'c matrix \cite{RS05,BGGC10}:
\begin{equation}
R_{ij}=\left\{
\begin{array}{ll}
(d_i d_j)^{-1/2} \quad &\mbox{if $v_i$ $\sim$ $v_j$},\\
0 \quad &\mbox{otherwise}.
\end{array}
\right.
\end{equation}

 \begin{figure}
 \centering
 \includegraphics[width=\textwidth]{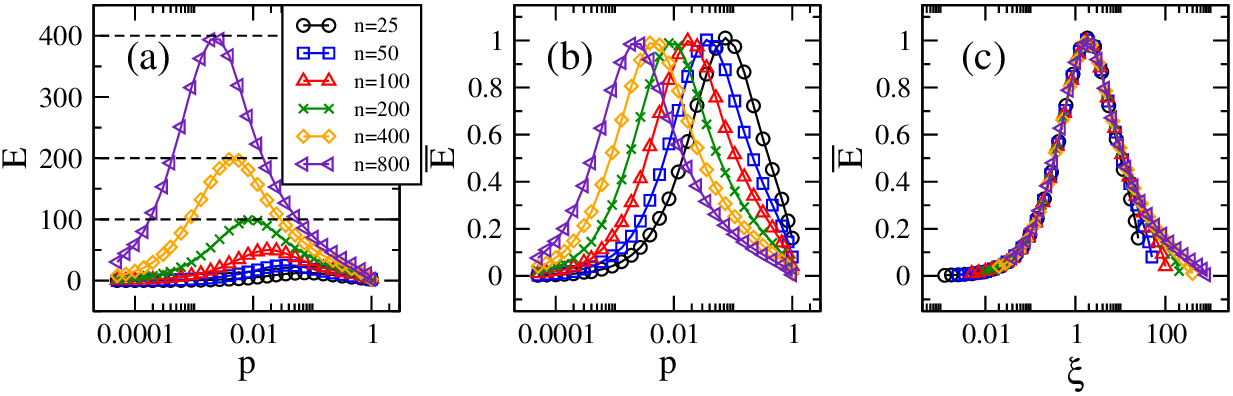}
\caption{(a) Randi\'c Matrix energy $E$ as a function of the probability $p$ for Erd\"os-R\'{e}nyi graphs of size 
$n$. Dashed lines indicate the values of $n/2$ for $n\in[200,800]$. (b) $\overline{E}=E/(n/2)$ as a function $p$. 
(c) $\overline{E}$ as a function $\xi$.}
\label{Fig3}
\end{figure}

Thus in Fig.~\ref{Fig3}(a) we present the energy $E$ as a function of the probability $p$ of Erd\"os-R\'{e}nyi graphs 
of several sizes $n$. The curves $E$ vs.~$p$ show a similar behavior for different values of $n$: For small $p$, $E$ 
increases with $p$ until it reaches $n/2$ (the maximum value it can take), then $E$ decreases from its maximum 
by further 
increasing $p$ giving to the curves $E$ vs.~$p$ a bell-like shape in log scale. Now, for convenience, we normalize 
$E$ to $n/2$ (that we name $\overline{E}$) and plot it in Fig.~\ref{Fig3}(b). Here it is clear that the curves $\overline{E}$
vs.~$p$ are very similar but shifted to the left on the $p$-axis for increasing $n$. Finally, in Fig.~\ref{Fig3}(c) we 
plot $\overline{E}$ as a function of the scaling parameter $\xi$, see Eq.~(\ref{xi}), and show that all curves fall one 
on top of the other
(except for finite size effects at large $\xi$). Therefore, we confirm that the energy of Erd\"os-R\'{e}nyi graphs (as
defined in Eq.~(\ref{energy})) also scales with the parameter $\xi$; that is, once $\xi$ is fixed the normalized energy 
$\overline{E}$ is (statistically) the same for different parameter combinations $(n,p)$.
Additionally, from Fig.~\ref{Fig3}(c) we can conclude that the maximum value of $E$ occurs in the interval 
$1<\xi<2$, in close agreement with the delocalization transition value for the eigenvectors of Erd\"os-R\'{e}nyi graphs 
reported in \cite{MAMRP15,MF91,FM91b,EE92,E92} to be $\xi \approx 1.4$. 

Even though we have shown that $\xi$ scales both $\bra \overline{R}(G) \ket$ and  $\overline{E}$ reasonably well,
it is fair to say that there are additional quantities related to $\overline{R}(G)$ which are still size dependent for 
fixed $\xi$. See for example Fig.~\ref{Fig4}, where we show probability distribution functions of $\overline{R}(G)$
at fixed $\xi$. In this figure we observe that, even for fixed $\xi$ (or equivalently, for fixed $\bra \overline{R}(G) \ket$), 
$P(\overline{R}(G))$ becomes narrower for increasing $n$. This means that the variance and the minimal and
maximal values of $\overline{R}(G)$ change with $n$, as can be clearly seen in Fig.~\ref{Fig5}.
This motivate us to look for bounds and inequalities on the Randi\'c index on Erd\"os-R\'{e}nyi graphs, 
which is the main topic of the following Section. 

\begin{figure}
 \centering
 \includegraphics[width=0.95\textwidth]{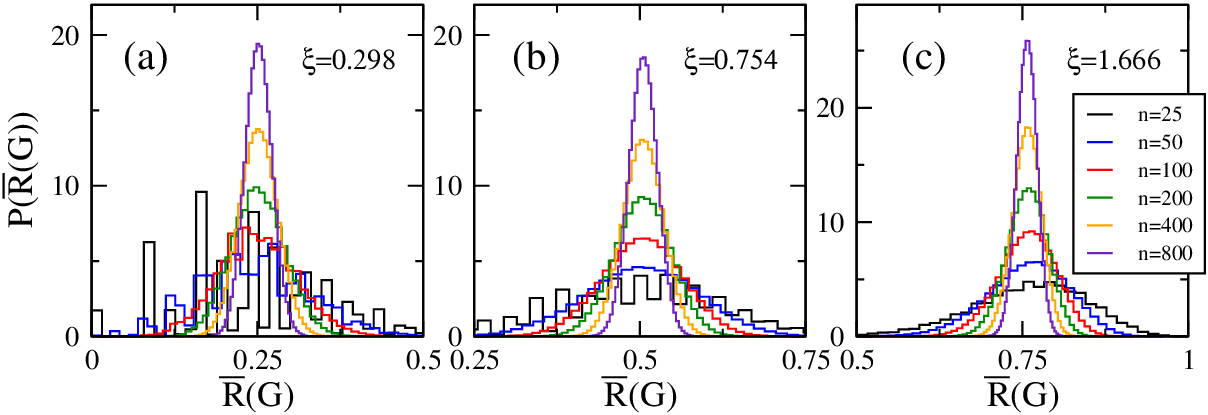}
\caption{Probability distribution functions of $\overline{R}(G)$, $P(\overline{R}(G))$, 
for several graph sizes $n$ at fixed values of $\bra \overline{R}(G) \ket$: 
(a) $\bra \overline{R}(G) \ket=0.25$, 
(b) $\bra \overline{R}(G) \ket=0.5$, and 
(c) $\bra \overline{R}(G) \ket=0.75$. 
The corresponding values of $\xi$ are given in the panels.
Each histogram was constructed with 2000 values of $\overline{R}(G)$.}
\label{Fig4}
\end{figure}
\begin{figure}
 \centering
 \includegraphics[width=0.95\textwidth]{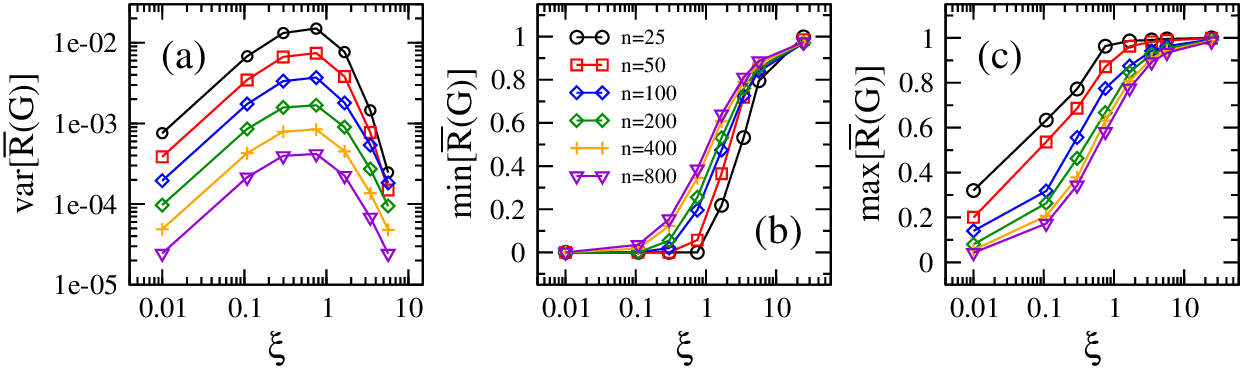}
\caption{(a) $\mbox{var}[\overline{R}(G)]$, (b) $\mbox{min}[\overline{R}(G)]$, and (c) $\mbox{max}[\overline{R}(G)]$
as a function of $\xi$. Each symbol was computed from 2000 values of $\overline{R}(G)$.}
\label{Fig5}
\end{figure}

\bigskip

\section{Inequalities for the Randi\'c index on Erd\"os-R\'{e}nyi models}
\label{Sec:Ana}

We recall that we consider a Random Graph $G$ from the standard Erd\"os-R\'{e}nyi model $G(n, p)$.
In the following, $G$ denotes a finite simple graph such that each connected component of $G$ has, at least, 
one edge (there are no isolated vertices).
We say that a statement holds for \emph{almost every graph} if the probability of the set of graphs for which the statement fails tends to $0$ as $n \to \infty$.

The following facts about the Erd\"os-R\'{e}nyi model are well-known \cite{Bollobas} (see also \cite{Dalfo}):
\begin{itemize}
\item[(1)] Almost every graph $G$ has $m = p\,n(n-1)/2 + o\,(n^2)$ edges.
\item[(2)] Almost every graph $G$ has maximum degree $\D = p(n - 1) + (2pqn \log n)^{1/2} + o\,((n \log n)^{1/2})$,
with $p \in [1/2,1)$ and $q = 1 - p$.
\item[(3)] Almost every graph $G$ has minimum degree $\d = q(n - 1) - (2pqn \log n)^{1/2} + o\,((n \log n)^{1/2})$,
with $p \in [1/2,1)$ and $q = 1 - p$.
\end{itemize}

In the previous equalities we are using Landau's notation: Recall that $f(n)=g(n)+o(a(n))$ means that
$$
\lim_{n \to \infty} \frac{f(n)-g(n)}{a(n)} = 0,
$$
and $f(n)=g(n)+O(a(n))$ means that
$$
\frac{f(n)-g(n)}{a(n)}
$$
is a bounded sequence.

The following result relates the Randi\'c and the $(-2)$-sum-connectivity indices. 


\begin{theorem} \label{t:gam}

Let $G$ be a graph with minimum degree $\d$ and maximum degree $\D$. Then

\begin{equation} \label{eq:gam}
\begin{aligned}
&4 \,\d \, \chi_{_{-2}}(G)\le R(G) \le 4 \D \, \chi_{_{-2}}(G),\qquad & \mbox{if} \ \ \d/\D \ge t_0, \nonumber \\
&4 \,\d \, \chi_{_{-2}}(G) \le R(G) \le \frac{(\D+\d)^2}{\sqrt{ \D\d}} \, \chi_{_{-2}}(G),\qquad & \mbox{if} \ \ \d/\D < t_0,
\end{aligned}
\end{equation}
where $t_0$ is the unique solution of the equation $t^3+5t^2+11t-1=0$ in the interval $(0,1)$.
The equality in the lower bound is attained if and only if $G$ is regular.
The equality in the first upper bound is attained if and only if $G$ is regular;
the equality in the second upper bound is attained if and only if $G$ is a biregular graph.
\end{theorem}

\begin{proof}
Since $a(t)=t^3+5t^2+11t-1$ is an increasing function on the interval $[0,1]$, $a(0)<0$ and $a(1)>0$,
there exists a unique solution of the equation $t^3+5t^2+11t-1=0$ in the interval $(0,1)$.
Hence, the number $t_0$ is well-defined.

Let us compute the maximum and minimum values of the function
$g: [\d,\D] \times [\d,\D] \to \RR$
given by
$$
g(x,y)
= \frac{\sqrt{xy}}{(x+y)^2} \,.
$$
\indent
Since $g(x,y)=g(y,x)$, we can assume that $x \le y$.
The partial derivatives of $g$ are
\begin{equation*}
\begin{split}
\frac{\partial g}{\partial x}(x,y)
& = \frac{x^{-1/2}y^{1/2} (x+y) -4x^{1/2}y^{1/2}}{2(x+y)^3}
\\
& = x^{-1/2}y^{1/2} \, \frac{ y -3x }{2(x+y)^3} \, ,
\\
\frac{\partial g}{\partial y}(x,y)
& = y^{-1/2}x^{1/2} \,\frac{ x -3y }{2(x+y)^3} \,.
\end{split}
\end{equation*}
\indent
Since $y \ge x \ge \d > 0$, we obtain $\partial g/\partial y<0$ and $g$ is a decreasing function on $y$.
Therefore, $g$ attains its minimum value on $\{ (x,\D) |\;\d\le x \le \D \}$
and its maximum value on $\{ (x,x) |\;\d\le x \le \D \}$.
Note that $g(x,x)=1/(4x) \le 1/(4\d)$.

If $\D \le 3\d$, then $\partial g/\partial x(x,\D) < 0$ for every $x > \d$.

If $\D > 3\d$, then $\partial g/\partial x(x,\D) > 0$ for every $\d \le x < \D/3$ and $\partial g/\partial x(x,\D) < 0$ for every $\D/3 < x \le \D$.

Hence, we have in every case
\begin{equation*}
\begin{split}
\min \Big\{ \, \frac{1}{4 \D} \,,\, \frac{\sqrt{ \D\d}}{(\D+\d)^2} \, \Big\}
= \min \big\{ g(\D, \D) ,\, g(\d,\D) \big\}
\le g(x,y)
\le \frac{1}{4 \d}\,.
\end{split}
\end{equation*}
Thus,
\begin{equation*}
\begin{aligned}
&\min \Big\{ \, \frac{1}{4 \D} \,,\, \frac{\sqrt{ \D\d}}{(\D+\d)^2} \, \Big\}\, \frac{1}{\sqrt{ d_u d_v}} \le \frac{1}{(d_u+d_v)^2} \le \frac{1}{4 \d}\, \frac{1}{\sqrt{ d_u d_v}} \,, \\
&\min \Big\{ \, \frac{1}{4 \D} \,,\, \frac{\sqrt{ \D\d}}{(\D+\d)^2} \, \Big\} R(G)\le \chi_{_{-2}}(G)\le \frac{1}{4 \d}\, R(G) .
\end{aligned}
\end{equation*}

\indent

If the equality in the lower bound  is attained, then
$(d_u,d_v)=(\d,\d)$ for all $uv \in E(G)$;
hence, $d_u = \d$ for all $u \in V(G)$ and so, $G$ is regular.

In order to prove the upper bounds, it suffices to show that the inequality
\begin{equation}
\label{eq:ff}
\frac{1}{4 \D} \le \frac{\sqrt{ \D\d}}{(\D+\d)^2}
\end{equation}
holds if and only if $\d/\D \ge t_0$.

Inequality (\ref{eq:ff}) is equivalent to the following statements
\begin{equation*}
\begin{aligned}
&(\D+\d)^2 \le 4 \D \sqrt{ \D\d} \,,
& \qquad
&\Big(1+\displaystyle \frac{\d}{\D}\Big)^2 \le 4 \sqrt{\displaystyle \frac{\d}{\D}} \,,\\
&\Big(1+\displaystyle \frac{\d}{\D}\Big)^4 \le 16 \,\displaystyle \frac{\d}{\D} \,,
& \qquad
&\displaystyle \frac{\d^4}{\D^4} + 4\,\displaystyle \frac{\d^3}{\D^3} + 6\,\displaystyle \frac{\d^2}{\D^2} -12\,\displaystyle \frac{\d}{\D} + 1 \le 0 .
\end{aligned}
\end{equation*}
Since $0 < \d/\D \le 1$, let us consider the function $b(t)=t^4 + 4t^3 + 6t^2 -12t + 1$ for $t \in (0,1]$.
Since $b(t)= (t-1)(t^3+5t^2+11t-1)= (t-1)\,a(t)$, we have $a(t) \le 0$ if and only if $t \in [t_0,1]$.
Hence, inequality (\ref{eq:ff}) holds if and only if $\d/\D \ge t_0$.
Since the coefficients of the polynomial $a(t)=t^3+5t^2+11t-1$ are rational numbers, and the coefficients of $t^3$ and $t^0$ of the polynomial $a(t)$ are $1$ and $-1$, respectively, we have that $t_0 \notin \QQ$.
Note that this condition is equivalent to $\d/\D > t_0$, since $t_0 \notin \QQ$;
therefore, the equality in (\ref{eq:ff}) is attained if and only if $\d=\D$.

Therefore, the upper bounds hold.

If $\d/\D \ge t_0$, then the previous argument gives that $f$ attains its minimum value just at the point $(\D,\D)$.
Thus, the equality in the upper bound is attained if and only if $(d_u,d_v)=(\D,\D)$ for every $uv \in E(G)$, i.e., $G$ is regular.

If $\d/\D < t_0$, then $f$ attains its minimum value just at the points $(\d,\D)$ and $(\D,\d)$.
Hence, the equality in the upper bound is attained if and only if $\{d_u,d_v\}=\{\d,\D\}$ for every $uv \in E(G)$, i.e., $G$ is biregular.
In this case, $G$ can not be regular since $\d < t_0\D < \D$.
\end{proof}

\smallskip

Theorem \ref{t:gam} have the following consequence on Random Graphs.

\begin{corollary} \label{c:gam}
In the Erd\"os-R\'{e}nyi model $G(n, p)$, with $p \in [1/2,1)$ and $q=1-p$, almost every graph $G$ satisfies
$$
4 qn + O\,((n \log n)^{1/2})
\le \frac{R(G)}{\chi_{_{-2}}(G)}
\le \max \Big\{ 4 p \,,\, \frac{1}{\sqrt{pq}} \, \Big\}\, n + O\,((n \log n)^{1/2}).
$$
\end{corollary}

\begin{proof}
The conclusion in Theorem \ref{t:gam} can be written as follows:
\begin{equation} \label{eq:c:gam}
4 \,\d
\le \frac{R(G)}{\chi_{_{-2}}(G)}
\le \max \Big\{ 4 \D \,,\, \frac{(\D+\d)^2}{\sqrt{ \D\d}} \, \Big\} .
\end{equation}
Thus, the first inequality is a direct consequence of (\ref{eq:c:gam}) and $(3)$.
Let us prove the second one.
Items $(2)$ and $(3)$ give for almost every graph
\begin{equation*}
\begin{split}
\frac{(\D+\d)^2}{\sqrt{ \D\d}}
& = \frac{\big( n+o\,((n \log n)^{1/2}) \big)^2}{\sqrt{ pqn^2+ O\,(n( n\log n)^{1/2})}}
= \frac{ n^2+o\,(n (n \log n)^{1/2})}{\sqrt{ pq}\, n + O\,(( n\log n)^{1/2})}
\\
& = \frac{ n}{\sqrt{ pq}}\Big( 1- \frac{O\,(( n\log n)^{1/2})}{\sqrt{ pq}\, n} \Big) + \displaystyle\frac{ o\,(n (n \log n)^{1/2})}{\sqrt{ pq}\, n + O(( n\log n)^{1/2})}
\\
& = \frac{ n}{\sqrt{ pq}} + O\,((n \log n)^{1/2}) + o\,( (n \log n)^{1/2})
\\
& =\frac{ n}{\sqrt{ pq}} + O\,((n \log n)^{1/2}) .
\end{split}
\end{equation*}

This fact, (\ref{eq:c:gam}) and item $(2)$ give the second inequality for almost every graph.
\end{proof}

\medskip

Corollary \ref{c:gam} has the following consequence.

\begin{corollary} \label{c:gam2}
In the Erd\"os-R\'{e}nyi model $G(n, p)$, with $p = 1/2$, almost every graph $G$ satisfies
$$
\frac{R(G)}{\chi_{_{-2}}(G)}
= 2 n + O\,((n \log n)^{1/2}).
$$
\end{corollary}

The following technical result appears in \cite[Corollary 2.3]{RS2}.

\begin{lemma} \label{c:t}
Let $g$ be the function $g(x,y)=2\sqrt{xy}/(x + y)$ with $0<a\le x,y \le b$. Then
$$
\frac{2\sqrt{ab}}{a + b} \le g(x,y) \le 1.
$$
\end{lemma}

Given a graph $G$, let us define
$$
\d_G = \min_{uv\in E(G)} \frac{2\sqrt{d_u d_v}}{d_u + d_v} \,,
\qquad \quad
\D_G = \min_{uv\in E(G)} \frac{2\sqrt{d_u d_v}}{d_u + d_v} \,.
$$
Let $G$ be a graph with maximum degree $\D$ and minimum degree $\d$. Then Lemma \ref{c:t} gives, for every 
$uv\in E(G)$,
\begin{equation} \label{eq:c:t}
\frac{2\sqrt{\D \d}}{\D + \d} \le \d_G \le \frac{2\sqrt{d_u d_v}}{d_u + d_v} \le \D_G \le 1.
\end{equation}

Since
$$
\d_G \le \frac{2 d_u d_v}{(d_u + d_v)\sqrt{d_u d_v}} \le \D_G
$$
for every $uv\in E(G)$, we obtain
\begin{equation} \label{eq:Dal}
\frac{\d_G}{2}\, \frac{d_u + d_v}{d_u d_v} \le \frac{1}{\sqrt{d_u d_v}} \le \frac{\D_G}{2}\, \frac{d_u + d_v}{d_u d_v} \,.
\end{equation}
For every function $f$, we have
$$
\sum_{uv\in E(G)} \big( f(d_u)+f(d_u) \big)
= \sum_{u\in V(G)} d_u f(d_u),
$$
and so,
$$
\sum_{uv\in E(G)} \frac{d_u + d_v}{d_u d_v}
= \sum_{uv\in E(G)} \Big( \frac{1}{d_u} + \frac{1}{d_v} \Big)
= \sum_{u\in V(G)} d_u \,\frac{1}{d_u}
= \sum_{u\in V(G)} 1
= n.
$$
This equality and (\ref{eq:Dal}) give the inequalities:
$$
\frac{n \d_G}{2} \le R(G) \le \frac{n \D_G}{2} \,.
$$

A similar result is proved in \cite{Dalfo}; there, the author uses an argument based on differential calculus.

As a consequence of the previous result and (\ref{eq:c:t}), we obtain the known inequalities
\begin{equation} \label{eq:c:Dal}
\frac{\sqrt{ \D\d}}{\D+\d} \, n \le R(G) \le \frac{n}{2} \,.
\end{equation}
Notice that the right inequality has already been computationally verified in Fig.~\ref{Fig1}.

\begin{proposition} \label{c:Dal}
In the Erd\"os-R\'{e}nyi model $G(n, p)$, with $p \in [1/2,1)$ and $q = 1 - p$, almost every graph $G$ satisfies
$$
R(G)
\geq \sqrt{ pq} \,n + O\,((n \log n)^{1/2})\,.
$$
\end{proposition}

\begin{proof}
Let us consider the Erd\"os-R\'{e}nyi model $G(n, p)$.
Almost every graph $G$ satisfies
\begin{equation*}
\begin{split}
\frac{\sqrt{ \D\d}}{\D+\d} \, n
& = \frac{\sqrt{ pqn^2+ O\,(n( n\log n)^{1/2})}}{ n+o\,((n \log n)^{1/2})} \, n
= \frac{\sqrt{ pq}\, n + O\,(( n\log n)^{1/2})}{ n+o\,( (n \log n)^{1/2})} \, n
\\
& = \sqrt{ pq} \,n \Big( 1- \frac{o\,(( n\log n)^{1/2})}{n} \Big) + \frac{ O(( n\log n)^{1/2})}{ n+o\,( (n \log n)^{1/2})}\, n
\\
& = \sqrt{ pq} \,n + o\,( (n \log n)^{1/2}) + O\,((n \log n)^{1/2})
\\
& = \sqrt{ pq} \,n + O\,((n \log n)^{1/2}) .
\end{split}
\end{equation*}
This fact and (\ref{eq:c:Dal}) allow to obtain the result.
\end{proof}

\begin{corollary} \label{c:Dal2}
In the Erd\"os-R\'{e}nyi model $G(n, p)$, with $p = 1/2$, almost every graph $G$ satisfies
$$
R(G)
= \frac{n}{2} + O\,((n \log n)^{1/2}) .
$$
\end{corollary}
In fact, this Corollary has already been computationally verified in Fig.~\ref{Fig1}.

\begin{proposition} \label{p:11}

Let $G$ be a graph with $n$ vertices, minimum degree $\d$ and maximum degree $\D$. Then
$$
\frac{n}{2} - \frac{1}{2\d^2} \big( M_1(G) - 2 M_2^{1/2}(G) \big)
\le R(G) \le \frac{n}{2} - \frac{1}{2\D^2} \big( M_1(G) - 2 M_2^{1/2}(G) \big) ,
$$
$$
 \frac{1}{2\D^2} \big( M_1(G) + 2 M_2^{1/2}(G)\big)- \frac{n}{2}
\le R(G) \le \frac{1}{2\d^2} \big( M_1(G) + 2 M_2^{1/2}(G) \big) - \frac{n}{2}.
$$
The equality is attained in each bound if and only if $G$ is a regular graph.
\end{proposition}

\begin{proof}
In the argument in the proof of \cite[Theorem 1]{DBG} appears the following relation:
\begin{equation} \label{eq:p11}
R(G)
= \frac{n}2 - \frac{1}2 \!\! \sum_{uv \in E(G)} \!\! \frac{\big( \sqrt{d_u} - \sqrt{d_v} \,\big)^2}{d_ud_v} \,,
\end{equation}
and we deduce
$$
\frac{n}2 - \frac{1}{2\d^2} \!\! \sum_{uv \in E(G)} \!\! \big( \sqrt{d_u} - \sqrt{d_v} \,\big)^2
\le
R(G)
\le \frac{n}2 - \frac{1}{2\D^2} \!\! \sum_{uv \in E(G)} \!\! \big( \sqrt{d_u} - \sqrt{d_v} \,\big)^2 .
$$
Since
$$
\sum_{uv \in E(G)} \!\! \big( \sqrt{d_u} - \sqrt{d_v} \,\big)^2
= \sum_{uv \in E(G)} \!\! ( d_u + d_v) - 2  \!\!\! \sum_{uv \in E(G)} \!\!\sqrt{d_ud_v}
= M_1(G) - 2 M_2^{1/2}(G),
$$
we obtain the first and second inequalities.

\medskip

Since
\begin{equation*}
\begin{split}
-n + \sum_{uv \in E(G)} \!\! & \frac{\big( \sqrt{d_u} + \sqrt{d_v} \,\big)^2}{d_ud_v}
= -\sum_{u \in V(G)} \!\! d_u\,\frac{1}{d_u} +\sum_{uv \in E(G)} \!\! \frac{\big( \sqrt{d_u} + \sqrt{d_v} \,\big)^2}{d_ud_v}
\\
& = -\sum_{uv \in E(G)} \!\! \Big( \frac{1}{d_u} + \frac{1}{d_v} \,\Big) +\sum_{uv \in E(G)} \!\! \frac{d_u +d_v + 2\sqrt{d_ud_v} }{d_ud_v}
\\
& = -\sum_{uv \in E(G)} \!\! \frac{d_u +d_v}{d_ud_v}  +\sum_{uv \in E(G)} \!\! \frac{d_u +d_v + 2\sqrt{d_ud_v} }{d_ud_v}
\\
& = \sum_{uv \in E(G)} \!\! \frac{2\sqrt{d_ud_v}}{d_ud_v}
= 2\,R(G),
\end{split}
\end{equation*}

we have
$$ R(G) = -\frac{n}2 + \frac{1}2 \!\!\sum_{uv \in E(G)} \!\! \frac{\big( \sqrt{d_u} + \sqrt{d_v} \,\big)^2}{d_ud_v} \,,$$
$$
-\frac{n}2 + \frac{1}{2\D^2} \!\! \sum_{uv \in E(G)} \!\! \big( \sqrt{d_u} + \sqrt{d_v} \,\big)^2
\le
R(G) \le -\frac{n}2 + \frac{1}{2\d^2} \!\! \sum_{uv \in E(G)} \!\! \big( \sqrt{d_u} + \sqrt{d_v} \,\big)^2 .
$$
Since
$$
\sum_{uv \in E(G)} \!\! \big( \sqrt{d_u} + \sqrt{d_v} \,\big)^2
= \sum_{uv \in E(G)} \!\! ( d_u + d_v) + 2  \!\!\! \sum_{uv \in E(G)} \!\!\sqrt{d_ud_v}
= M_1(G) + 2 M_2^{1/2}(G),
$$
we obtain the third and forth inequalities.

\smallskip

If $G$ is a regular graph, then $\d=\D$ and, in each line, the lower and upper bounds are the same, and they are equal to $R(G)$.

If the equality is attained in some bound, then we have either
$d_ud_v=\d^2$ for every $uv \in E(G)$
or $d_ud_v=\D^2$ for every $uv \in E(G)$.
Thus, we have either
$d_u=\d$ for every $u \in V(G)$
or $d_u=\D$ for every $u \in V(G)$,
and so, the graph is regular.
\end{proof}

\medskip

Proposition \ref{p:11} has the following consequence on random graphs.

\begin{corollary} \label{c:11}
In the Erd\"os-R\'{e}nyi model $G(n, p)$, with $p \in [1/2,1)$ and $q = 1 - p$, almost every graph $G$ satisfies
$$
q^2n^2 + O\,(n^{3/2}( \log n)^{1/2})
\le \frac{M_1(G) - 2 M_2^{1/2}(G)}{n-2R(G)}
\le p^2n^2 + O\,(n^{3/2}( \log n)^{1/2}) ,
$$
$$
q^2n^2 + O\,(n^{3/2}( \log n)^{1/2})
\le \frac{M_1(G) + 2 M_2^{1/2}(G)}{n+2R(G)}
\le p^2n^2 + O\,(n^{3/2}( \log n)^{1/2}) .
$$
\end{corollary}

\begin{proof}
Proposition \ref{p:11} gives
\begin{equation*}
\begin{aligned}
&\d^2
& \le  \frac{M_1(G) - 2 M_2^{1/2}(G)}{n-2R(G)}
\le \D^2 ,
\\
&\d^2
& \le  \frac{M_1(G) + 2 M_2^{1/2}(G)}{n+2R(G)}
\le \D^2 .
\end{aligned}
\end{equation*}
Items $(2)$ and $(3)$ give for almost every graph
\begin{equation*}
\begin{aligned}
&\D^2
& = \big( pn + O\,(( n\log n)^{1/2}) \big)^2
= p^2n^2 + O\,(n^{3/2}( \log n)^{1/2}) ,
\\
&\d^2
& = \big( qn + O\,(( n\log n)^{1/2}) \big)^2
= q^2n^2 + O\,(n^{3/2}( \log n)^{1/2}) .
\end{aligned}
\end{equation*}

These facts give the desired inequalities.
\end{proof}

\medskip

The following proposition is a consequence of (\ref{eq:p11}) in \cite{DBG}.

\begin{proposition} \label{p:111}
Let $G$ be a graph with  $m$ edges, $n$ vertices, minimum degree $\d$ and maximum degree $\D$. Then  then
$$
R(G) \ge \frac{n}{2} - \frac{m}{2} \Big( \frac{1}{\sqrt{\d}} - \frac{1}{\sqrt{\D}} \Big)^2,
$$
and the equality is attained if and only if $G$ is a regular or biregular graph.
\end{proposition}

\begin{proof}
Equation (\ref{eq:p11}) can be written as
$$
R(G)
= \frac{n}2 - \frac{1}2\sum_{uv \in E(G)} \!\! \Big( \frac{1}{\sqrt{d_u}} - \frac{1}{\sqrt{d_v}} \Big)^2,
$$
and so,
$$
R(G)
\ge \frac{n}2 - \frac{1}2\sum_{uv \in E(G)} \!\! \Big( \frac{1}{\sqrt{\d}} - \frac{1}{\sqrt{\D}} \Big)^2
= \frac{n}{2} - \frac{m}{2} \Big( \frac{1}{\sqrt{\d}} - \frac{1}{\sqrt{\D}} \Big)^2.
$$

The equality is attained if and only if
$\{d_u,d_v\}=\{\d,\D\}$ for every $uv \in E(G)$, i.e.,
$G$ is a regular or biregular graph.
\end{proof}

\medskip

Note that the lower bound in Proposition \ref{p:111} is not comparable with the one in Corollary \ref{c:Dal}, as the following examples show:

If $G$ is the path graph with $n$ vertices, then
$$
\frac{n}{2} - \frac{m}{2} \Big( \frac{1}{\sqrt{\d}} - \frac{1}{\sqrt{\D}} \Big)^2
=\frac{n}{2} - \frac{n-1}{2} \Big( 1 - \frac{1}{\sqrt{2}} \Big)^2
\approx \frac{2\sqrt{2}+1}4 \, n
$$
is larger than
$$
\frac{\sqrt{ \D\d}}{\D+\d} \, n
=\frac{\sqrt{ 2}}{3} \, n ,
$$
for large enough $n$.
However, if $G$ is the complete graph with $n-1$ vertices $K_{n-1}$ with an additional edge joining a vertex of $K_{n-1}$ with an additional vertex of degree $1$, then
$$
\frac{\sqrt{ \D\d}}{\D+\d} \, n
=\frac{\sqrt{ n-1}}{n} \, n
= \sqrt{ n-1}
$$
is larger than
$$
\frac{n}{2} - \frac{m}{2} \Big( \frac{1}{\sqrt{\d}} - \frac{1}{\sqrt{\D}} \Big)^2
=\frac{n}{2} - \frac{\frac12 (n-1)(n-2) + 1}{2} \Big( 1 - \frac{1}{\sqrt{n-1}} \Big)^2 \,,
$$
for large enough $n$.

\begin{proposition} \label{c:111}
In the Erd\"os-R\'{e}nyi model $G(n, p)$, with $p \in [1/2,1)$ and $q = 1 - p$, almost every graph $G$ satisfies
\begin{equation*}
\begin{aligned}
&R(G) \ge \frac{2\sqrt{pq} + 2q-1}{4q} \, n + o\,( n), & \qquad \mbox{ if } \ \  p>1/2,\\
&R(G)= \frac{n}{2} + O\,((n \log n)^{1/2}) , & \qquad \mbox{ if } \ \ p=1/2.
\end{aligned}
\end{equation*}
\end{proposition}

\begin{proof}
The second statement follows from Corollary \ref{c:Dal2}.

Assume now $p>1/2$.
Items $(2)$ and $(3)$ give that in the Erd\"os-R\'{e}nyi model $G(n, p)$,
almost every graph $G$ satisfies
\begin{equation*}
\begin{split}
\Big( \frac{1}{\sqrt{\d}} - \frac{1}{\sqrt{\D}} \Big)^2
& = \frac{\D+\d-2\sqrt{\D\d}}{\D\d} 
\\
& = \frac{n + o\,( (n \log n)^{1/2})-2\sqrt{pq n^2 + O\,(n( n\log n)^{1/2})}}{pq n^2 + O\,(n( n\log n)^{1/2})}
\\
& = \frac{n + o\,( (n \log n)^{1/2})-2\sqrt{pq}\, n + O\,(( n\log n)^{1/2})}{pq n^2 + O\,(n( n\log n)^{1/2})}
\\
& = \frac{\big(1-2\sqrt{pq}\,\big)n + O\,(( n\log n)^{1/2})}{pq n^2 + O\,(n( n\log n)^{1/2})}
\\
& = \frac{1-2\sqrt{pq}}{pq n} + o\,\Big( \frac1n \Big) .
\end{split}
\end{equation*}
This fact, Proposition \ref{p:111} and item $(1)$ give
\begin{equation*}
\begin{split}
R(G)
& \ge  \frac{n}2 - \frac{m}2 \Big( \frac{1}{\sqrt{\d}} - \frac{1}{\sqrt{\D}} \Big)^2
\\
& = \frac{n}2 - \frac12 \Big( \frac{p n(n-1)}{2} + o\,\big( n^2 \big) \Big) \Big( \frac{1-2\sqrt{pq}}{pq n} + o\,\Big( \frac1n \Big) \Big)
\\
& = \frac{1}2\,n - \frac{1-2\sqrt{pq}}{4q} \, n + o\,( n)
= \frac{2\sqrt{pq} + 2q-1}{4q} \, n + o\,( n) .
\end{split}
\end{equation*}
\end{proof}

The \emph{misbalance rodeg index} is defined as
$$
M\!R\,(G) = \sum_{uv\in E(G)} \big| \sqrt{d_u} - \sqrt{d_v} \, \big| .
$$
This is a significant predictor of enthalpy of vaporization
and of standard enthalpy of vaporization for octane
isomers (see \cite{V1}).

\begin{theorem} \label{t:13}
Let $G$ be a graph with maximum degree $\D$ and $m$ edges. Then
$$
R(G)
\le \frac{n}2 - \frac1{2\D^2 m} \, M\!R (G)^2,
$$
and the equality is attained if and only if $G$ is regular.
\end{theorem}

\begin{proof}
By Cauchy-Schwarz inequality we have
\begin{equation*}
\begin{split}
M\!R\,(G)^2 & =
\Big( \sum_{uv \in E(G)} \!\! \big| \sqrt{d_u} - \sqrt{d_v} \,\big|\Big)^2
\\
& \le \Big( \sum_{uv \in E(G)} \!\! 1^2 \Big)
\sum_{uv \in E(G)} \!\! \big( \sqrt{d_u} - \sqrt{d_v} \,\big)^2
\le m \!\!\!\! \sum_{uv \in E(G)} \!\! \big( \sqrt{d_u} - \sqrt{d_v} \,\big)^2.
\end{split}
\end{equation*}
Hence, (\ref{eq:p11}) gives
\begin{equation*}
\begin{split}
R(G)
& = \frac{n}2 - \frac{1}2\sum_{uv \in E(G)} \!\! \frac{\big( \sqrt{d_u} - \sqrt{d_v} \,\big)^2}{d_ud_v}
\\
& \le \frac{n}2 - \frac{1}{2\D^2}\sum_{uv \in E(G)} \!\! \big( \sqrt{d_u} - \sqrt{d_v} \,\big)^2
\le \frac{n}2 - \frac1{2\D^2 m} \, M\!R\,(G)^2 .
\end{split}
\end{equation*}
\indent
If $G$ is regular, then $R(G) = n/2$ and $M\!R\,(G) = 0$ and so, the equality is attained.

If the equality is attained, then $d_u d_v = \D^2$ for every $uv \in E(G)$;
thus, $d_u = \D$ for all $u \in V(G)$ and so, $G$ is a regular graph.
\end{proof}

\begin{corollary} \label{c:13}
In the Erd\"os-R\'{e}nyi model $G(n, p)$, with $p \in [1/2,1)$ and $q = 1 - p$, almost every graph $G$ satisfies
$$
\frac{M\!R (G)^2}{n-2R(G)}
\le \frac12\,p^2 n^3 + o\,(n^3) .
$$
\end{corollary}

\begin{proof}
Theorem \ref{t:13} gives the inequality
$$
\frac{M\!R (G)^2}{n-2R(G)}
\le \D^2 m .
$$
Items $(1)$ and $(2)$ give that in the Erd\"os-R\'{e}nyi model $G(n, p)$,
almost every graph $G$ satisfies
\begin{equation*}
\begin{split}
\D^2 m
& = \big( p n + O\,(( n\log n)^{1/2}) \big)^2 \Big( \frac{p n(n-1)}{2} + o\,\big( n^2 \big) \Big)
\\
& = \frac12\,p^2 n^3 + o\,(n^3) ,
\end{split}
\end{equation*}
and this gives the desired inequality.
\end{proof}

\medskip

The following Sz\"okefalvi Nagy inequality appears in \cite{SN} (see also \cite{SGK}).

\begin{lemma} \label{l:Szokefalvi Nagy}
If $a_j \ge 0$ for $1\le j \le k$, $R= \max_j a_j$ and $r= \min_j a_j$, then
$$
k\sum_{j=1}^k a_j^2 - \Big( \sum_{j=1}^k a_j \Big)^2
\ge \frac{k}{2} \, (R - r)^2 .
$$
\end{lemma}

In many papers the hypothesis
$a_j \ge 0$ for $1\le j \le k$, $R= \max_j a_j$ and $r= \min_j a_j$,
is replaced by
$0 < r \le a_j \le R$ for $1\le j \le k$.
However, the conclusion of Lemma \ref{l:Szokefalvi Nagy} does not hold in general with the hypothesis $0 < r \le a_j \le R$ for $1\le j \le k$,
as the following example shows:

If $a_j = a$ for $1\le j \le k$, $R> a$ and $r\le a < R$, then
$$
k\sum_{j=1}^k a_j^2 - \Big( \sum_{j=1}^k a_j \Big)^2
= k^2 a^2 - k^2 a^2
= 0
< \frac{k}{2} \, (R - r)^2 .
$$

\begin{theorem} \label{t:Szokefalvi Nagy}
Let $G$ be a graph with $m$ edges,
$$
\Pi = \max_{uv \in E(G)} \frac{1}{\sqrt{d_u d_v}} \,,
\qquad and \qquad
\pi = \min_{uv \in E(G)} \frac{1}{\sqrt{d_u d_v}} \,.
$$
Then
$$
R(G)
\le \sqrt{m M_2^{-1}(G) - \frac{m}{2}\, (\Pi-\pi)^2 } \, ,
$$
and the equality is attained if $G$ is a regular or biregular graph.
\end{theorem}

\begin{proof}
If we choose $a_j=1/\sqrt{d_u d_v}\,$,
Lemma \ref{l:Szokefalvi Nagy} gives
\begin{equation*}
\begin{split}
m M_2^{-1}(G) - R(G)^2
& = m\sum_{uv \in E(G)} \frac{1}{d_u d_v} - \left( \sum_{uv \in E(G)} \frac{1}{\sqrt{d_u d_v}} \,\right)^2
\\
& \ge \frac{m}{2} \, (\Pi-\pi)^2 ,
\end{split}
\end{equation*}
and this gives the inequality.

\smallskip

If $G$ is a  biregular or regular graph, then
$$
\frac{1}{\sqrt{d_u d_v}}
= \frac{1}{\sqrt{\D\d}}
= \Pi
= \pi
$$
for every $uv \in E(G)$.
Thus,
$$
\sqrt{m M_2^{-1}(G) -\frac{m}{2} \, (P-p)^2 }
= \sqrt{m \,\frac{m}{\D\d}}
= \frac{m}{\sqrt{\D\d}}
= R(G) .
$$
\end{proof}

The \emph{inverse degree index} $ID(G)$ is defined by
$$
ID(G) = \sum_{u\in V(G)} \frac1{d_u}
= \sum_{uv\in E(G)} \Big( \frac1{d_u^2} + \frac1{d_v^2}\Big)
= \sum_{uv\in E(G)}\frac{d_u^2 + d_v^2}{d_u^2 d_v^2}\,.
$$
The inverse degree index of a graph has been studied by several authors (see, e.g., \cite{DH,DXW,EP} and the references therein). The following result provides some inequalities relating Randi\'c and Inverse Degree indices
(see \cite{RSS} for other inequalities relating these indices).

\begin{theorem} \label{t:rid}
Let $G$ be a graph with minimum degree $\d$ and maximum degree $\D$.
Then
\begin{equation*}
\begin{aligned}
&\frac{\d}{2} \, ID(G)
\le R(G)
\le \frac{\D}{2} \, ID(G),
& \qquad \mbox{if} \ \ \d \ge s_0\D,
\\
&\frac{(\D \d)^{3/2}}{\D^2+\d^2} \, ID(G)
\le R(G)
\le \frac{\D}{2} \, ID(G),
& \qquad \mbox{if} \ \ \d \le s_0\D,
\end{aligned}
\end{equation*}
where $s_0$ is the unique solution of the equation $s^2-2\sqrt{s} +1=0$ in $(0,1)$.
Furthermore, the upper bound is attained if and only if $G$ is regular;
if $\d \ge s_0\D$, then the lower bound is attained if and only if $G$ is regular;
if $\d \le s_0\D$, then the lower bound is attained if and only if $G$ is biregular.
\end{theorem}

\begin{proof}
First of all, let us check that $s_0$ is well-defined, i.e., there exists a unique solution of the equation $s^2-2\sqrt{s} +1=0$ in $(0,1)$.
By making the change of variable $s=t^2$, we see that this holds if and only if there exists a unique solution of the equation $t^4-2t +1=0$ in $(0,1)$.
Note that $t^4-2t +1=(t-1)u(t)$, with $u(t)=t^3+t^2+t-1$.
Since $u(0)=-1$, $u(1)=2$ and $u'(t)=3t^2+2t+1>0$ on $(0,1)$, we conclude that there is a unique zero $t_0$ of $u$ in $(0,1)$ and,
in fact,
$u(t)<0$ for every $t \in (0,t_0)$ and $u(t)>0$ for every $t \in (t_0,1)$.
If $s_0=t_0^2$, then $s^2-2\sqrt{s} +1>0$ for $s \in (0,s_0)$ and $s^2-2\sqrt{s} +1<0$ for every $s \in (s_0,1)$.

Let $f: [\d,\D] \times [\d,\D] \to \RR$ be the function
given by
$$
f(x,y)
= \Big( \frac1{x^2} + \frac1{y^2}\Big)\sqrt{xy}
= x^{-3/2} y^{1/2} + y^{-3/2} x^{1/2} .
$$
\indent

First we will find the minimum and maximum values of $f$. We can assume that $x \le y$(symmetry).
$$
\frac{\partial f}{\partial x}(x,y)
= -\frac{3}{2}\,x^{-5/2} y^{1/2} + \frac{1}{2}\,x^{-1/2} y^{-3/2}
= \frac{1}{2}\,x^{-5/2} y^{-3/2}\, ( x^2 - 3y^2 ).
$$
Thus,
$$
\frac{\partial f}{\partial x}(x,y) <0,
\qquad
\mbox{if} \ \ \d \le x \le y \le \D,
$$
and so, the function $f$ attains its maximum value in the set $\{x=\d,\, \d\le y \le \D \}$,
and the minimum value in the set $\{\d\le x=y \le \D \}$.
Thus,
\begin{equation*}
\begin{split}
f(x,y)
& \ge \min_{\d\le x \le \D} f(x,x)
= \min_{\d\le x \le \D} \frac2{x^2}\,x
= \frac2{\D}\, ,
\\
\frac1{d_u^2} + \frac1{d_v^2}
& \ge \frac2{\D} \,\frac1{\sqrt{d_ud_v}}\, ,
\\
R(G)
& \le \frac{\D}2 \, ID(G).
\end{split}
\end{equation*}
Since
$$
\frac{\partial f}{\partial y}(x,y)
= \frac{1}{2}\,y^{-5/2} x^{-3/2}\, ( y^2 - 3x^2 ),
$$
if $\D^2-3\d^2 < 0$, then
$$
\frac{\partial f}{\partial y}(\d,y)
= \frac{1}{2}\,y^{-5/2} \d^{-3/2}\, ( y^2 - 3\d^2 )
\le \frac{1}{2}\,y^{-5/2} \d^{-3/2}\, ( \D^2-3\d^2 ) < 0,
$$
and
$$
f(x,y)
\le \max_{\d\le y \le \D} f(\d,y)
= f(\d,\d)
= \frac2{\d}\,.
$$
If $\D^2-3\d^2 \ge 0$, then
$$
\frac{\partial f}{\partial y}(\d,y)
= \frac{1}{2}\,y^{-5/2} \d^{-3/2}\, ( y^2 - 3\d^2 ) \le 0
$$
if and only if $y \in [\d,\sqrt3 \, \d]$.
Thus, $f(\d,y)$ decreases on $[\d,\sqrt3 \, \d]$ and increases on $[\sqrt3 \, \d,\D]$.
Hence,
we have in both cases
$$
f(x,y)
 \le \max_{\d\le y \le \D} f(\d,y)
= \max \big\{ f(\d,\d), \, f(\d,\D)\big\}
= \max \Big\{ \frac2{\d}\,, \, \Big(\frac1{\d^2} + \frac1{\D^2}\Big)
\sqrt{\d\D}\, \Big\}
\, .
$$
Recall that $s^2-2\sqrt{s} +1>0$ on $(0,s_0)$.
Thus, we have for $\d \le s_0\D$,
$$
\Big(1+\frac{\d^2}{\D^2}\Big) \ge 2\sqrt{\frac{\d}{\D}}\,,
\qquad
\Big(\frac1{\d^2} + \frac1{\D^2}\Big)
\sqrt{\d\D} \ge \frac2{\d}\,,
$$
and we conclude
\begin{equation*}
\begin{split}
f(x,y)
& \le \max \Big\{ \frac2{\d}\,, \, \Big(\frac1{\d^2} + \frac1{\D^2}\Big)
\sqrt{\d\D}\, \Big\}
= \frac{\D^2+\d^2}{(\D \d)^{3/2}}
\, ,
\\
\frac1{d_u^2} + \frac1{d_v^2}
& \le \frac{\D^2+\d^2}{(\D \d)^{3/2}} \,\frac1{\sqrt{d_ud_v}}\, ,
\\
R(G)
& \ge \frac{(\D \d)^{3/2}}{\D^2+\d^2} \, ID(G).
\end{split}
\end{equation*}
If $\d \ge s_0\D$, then $f(x,y) \le f(\d,\d) = 2/\d$ and
$$
R(G)
\ge \frac{\d}{2} \, ID(G).
$$

The previous argument gives that the upper bound is attained if and only if $d_u=d_v=\D$ for every $uv \in E(G)$,
and this happens if and only if $G$ is regular.

Assume that $\d \ge s_0\D$.
Thus, the lower bound is attained if and only if $d_u=d_v=\d$ for every $uv \in E(G)$,
i.e., if and only if $G$ is regular.

Assume that $\d \le s_0\D$.
Thus, the lower bound is attained if and only if $\{d_u,d_v\} = \{\D,\d\}$ for every $uv \in E(G)$,
i.e., if and only if $G$ is biregular (note that $G$ can not be a regular graph since $\d \le s_0\D < \D$).
\end{proof}

\medskip

Theorem \ref{t:rid} has the following consequence on random graphs.

\begin{corollary} \label{c:rid}
In the Erd\"os-R\'{e}nyi model $G(n, p)$, with $p \in [1/2,1)$ and $q = 1 - p$, almost every graph $G$ satisfies
$$
\min \Big\{\frac{q}{2} \,, \, \frac{(pq)^{3/2}}{ p^2+q^2} \Big\} \, n
+ O\,(( n\log n)^{1/2})
\le \frac{R(G)}{ID(G)}
\le \frac{p}{2} \,n + O\,(( n\log n)^{1/2}) .
$$
\end{corollary}

\begin{proof}
Theorem \ref{t:rid} can be stated as follows:
$$
\min \Big\{\frac{\d}{2} \,, \, \frac{(\D \d)^{3/2}}{\D^2+\d^2} \, \Big\}
\le \frac{R(G)}{ID(G)}
\le \frac{\D}{2} \, .
$$
Items $(2)$ and $(3)$ give for almost every graph
\begin{equation*}
\begin{split}
\frac{(\D \d)^{3/2}}{\D^2+\d^2}
& = \frac{ \big(pqn^2 + O\,(n( n\log n)^{1/2}) \big)^{3/2}}{ (p^2+q^2)n^2 + O\,(n( n\log n)^{1/2}) }
= \frac{ (pq)^{3/2}n^3 \big(1 + \frac32\, \frac{O\,(n( n\log n)^{1/2})}{pqn^2} \big)}{ (p^2+q^2)n^2 + O\,(n( n\log n)^{1/2}) }
\\
& = \frac{ (pq)^{3/2}n^3}{ (p^2+q^2)n^2 + O\,(n( n\log n)^{1/2}) }
+ \frac{ O\,(n^2( n\log n)^{1/2})}{ (p^2+q^2)n^2 + O\,(n( n\log n)^{1/2}) }
\\
& = \frac{ (pq)^{3/2}n}{p^2+q^2} \Big( 1 - \frac{O\,(( n\log n)^{1/2})}{n} \Big)
+ O\,(( n\log n)^{1/2})
\\
& = \frac{(pq)^{3/2}}{ p^2+q^2} \,n
+ O\,(( n\log n)^{1/2}) .
\end{split}
\end{equation*}
These facts, and items $(2)$ and $(3)$ give for almost every graph
\begin{multline*}
\min \Big\{\frac{q}{2} \,n + O\,(( n\log n)^{1/2}) \,, \, \frac{(pq)^{3/2}}{ p^2+q^2} \, n
+ O\,(( n\log n)^{1/2}) \, \Big\}
\le \frac{H(G)}{ID(G)}\\
\le \frac{p}{2} \,n + O\,(( n\log n)^{1/2}) ,
\end{multline*}
and this finishes the proof.
\end{proof}

\section{Summary}

Based on the important theoretical-practical applications of the Randi\'c index, in this paper we have studied 
computationally and analytically the properties of the Randi\'c index $R(G)$ in Erd\"os-R\'{e}nyi graphs $G(n,p)$ 
characterized by $n$ vertices connected independently with probability $p \in (0,1)$.

First, by the proper scaling analysis of the average (and normalized) Randi\'c index, 
$\bra \overline{R}(G) \ket = \bra R(G)\ket/(n/2)$, we found that $\xi\approx np$ is the scaling parameter 
of $R(G(n,p))$; that is, for fixed $\xi$, $\bra \overline{R}(G) \ket$ is also fixed, see Fig.~\ref{Fig2}(b). 
Moreover, our analysis provides a way to predict the value of the Randi\'c index on Erd\"os-R\'{e}nyi graphs 
once the value of $\xi$ is known: 
$R(G)\approx 0$ for $\xi < 0.01$ (when the vertices in the graph are mostly isolated),
the transition from isolated vertices to complete graphs occurs in the interval $0.01 < \xi < 10$ where $0<R(G)< n/2$,
while when $\xi > 10$ the graphs are almost complete and $R(G)\approx n/2$. 
These intervals are indicated as vertical dashed lines in Fig.~\ref{Fig2}(b).
Also, to extend the applicability of our scaling analysis we demonstrate that for fixed $\xi$ the spectral properties 
of $R(G(n,p))$ (characterized by the energy of the corresponding Randi\'c matrix) are also universal; 
i.e., they do not depend on the specific values of the individual graph parameters, see Fig.~\ref{Fig3}(c). 

In particular, we would like to stress that here we have successfully introduced a scaling approach to the 
study of topological indexes. 

Then, to complement the study of the Randi\'c index we have explored the relations between $R(G)$ and 
other important topological indexes such as the (-2) sum-connectivity index, the misbalance rodeg index, the inverse
degree index, among others. In particular, we characterized graphs which are extremal with respect to those relations.


\section*{Acknowledgements}
C.T.M.-M. and J.A.M.-B. thank partial support by VIEP-BUAP (Grant No. MEBJ-EXC18-G), Fondo
Institucional PIFCA (Grant No. BUAP-CA-169), and CONACyT (Grant No. CB-2013/220624), Mexico.
J.M.R. and J.M.S. were supported in part by two grants from Ministerio de Econom{\'\i}a y Competitividad, Agencia Estatal de Investigaci{\'\o}n (AEI) and Fondo Europeo de Desarrollo Regional (FEDER) 
(MTM2016-78227-C2-1-P and MTM2017-90584-REDT), Spain.

\end{document}